\documentclass[11pt, a4paper]{article}
\usepackage[margin = 1in]{geometry}
\usepackage{libertinus}
\usepackage{amsmath}
\usepackage{amssymb}
\usepackage{amsfonts}
\usepackage{amsthm} 
\usepackage{graphicx}
\usepackage{comment}
\usepackage{thm-restate}
\usepackage{hyperref}
\usepackage{cleveref}
\usepackage{here}
\usepackage{cite}
\usepackage{enumerate}
\usepackage{tabularx}
\usepackage{environ}
\usepackage{url}
\usepackage[normalem]{ulem}
\usepackage{color}
\usepackage[T1]{fontenc}
\usepackage{booktabs}

\newtheorem{theorem}{Theorem}
\newtheorem{lemma}[theorem]{Lemma}

\newtheorem{claim}[theorem]{Claim}

\newtheorem{proposition}[theorem]{Proposition}

\usepackage{algorithm}
\usepackage[noend]{algpseudocode}

\algrenewcommand\algorithmicrequire{\textbf{Input:}}
\algnewcommand\algorithmiccontinue{\textbf{continue}}
\algnewcommand\algorithmicbreak{\textbf{break}}

\global\long\def\eps{\varepsilon}%
\global\long\def\epsilon{\varepsilon}%
\global\long\def\eps{\varepsilon}%

\newcommand{\bR}{\mathbb{R}}

\newcommand{\Poly}{\mathrm{poly}}

\newcommand{\ds}{\displaystyle}

\newcommand{\subimportant}[1]{\textbf{#1}}
\newcommand{\important}[1]{\textcolor{red}{\textbf{#1}}}

\ifdefined\ShowComment
\newcommand{\tatsuya}[1]{{\bf \color{red} tatsuya: #1}}
\else
\newcommand{\tatsuya}[1]{}
\fi

\title{Improved Pass Complexity for Semi-Streaming Matching \\ in General Graphs via Budgeted Price Approach}

\author{Tatsuya Terao\thanks{
Research Institute for Mathematical Sciences, Kyoto University.
E-mail: ttatsuya@kurims.kyoto-u.ac.jp
}
}

\date{}

\begin{document}

\maketitle

\begin{abstract}
We present a simple deterministic semi-streaming algorithm that computes a $(1-\varepsilon)$-approximate maximum matching in general graphs in $O(1/\varepsilon^{2})$ passes.
This matches the pass complexity of the auction algorithm of Assadi, Liu, and Tarjan~[SOSA'21] for bipartite graphs, and answers Open Question~2 of Assadi~[SOSA'24 \& TheoretiCS'25].

Our approach, which we call the \emph{budgeted price approach}, is very different from the previous $\mathrm{poly}(1/\varepsilon)$-pass semi-streaming algorithms for $(1-\varepsilon)$-approximation of matching in general graphs and improves their pass complexity.
All previous approaches are augmenting-path-based and use $O(1/\varepsilon^{19})$ passes by Fischer, Mitrovi\'{c}, and Uitto~[STOC'22], $O(1/\varepsilon^{6})$ passes by Mitrovi\'{c}, Mukherjee, Sankowski, and Sheu~[ICALP'25], and $O(1/\varepsilon^{4})$ passes by Izumi, Kitamura, and Yamaguchi~['25].
\end{abstract}

\section{Introduction} \label{sec:introduction}

In the semi-streaming model for graph problems, formalized by Feigenbaum, Kannan, McGregor, Suri, and Zhang~\cite{feigenbaum2005graph}, the edges of an $n$-vertex graph $G=(V,E)$ arrive in an arbitrarily ordered stream, and an algorithm makes a small number of passes over the stream using $O(n \cdot \Poly(\log n)$ space.
We study the maximum matching problem in this model.
The goal is to compute a $(1-\eps)$-approximate maximum matching in as few passes as possible, using $O(n \cdot \Poly(\log n, 1/\eps))$ space.

The first $(1-\eps)$-approximation semi-streaming algorithm, due to McGregor~\cite{mcgregor2005finding}, is randomized and uses $(1/\eps)^{O(1/\eps)}$ passes.
Tirodkar~\cite{tirodkar2018deterministic} later gave a deterministic algorithm with the same pass complexity.

For bipartite graphs, Eggert, Kliemann, and Srivastav~\cite{eggert2009bipartite} gave a deterministic algorithm with $O(1/\eps^{8})$ passes, which was improved to $O(1/\eps^{5})$ passes by Eggert, Kliemann, Munstermann, and Srivastav~\cite{eggert2012bipartite}.
Ahn and Guha~\cite{ahn2013linear} obtained $O(\log(1/\eps)/\eps^{2})$ passes even for weighted bipartite graphs, and the auction algorithm of Assadi, Liu, and Tarjan~\cite{assadi2021auction} achieves $O(1/\eps^{2})$ passes.
\subimportant{Their auction algorithm computes a sequence of $O(1/\eps^{2})$ maximal matchings in subgraphs of the input graph}, and hence can be implemented in the semi-streaming model with $O(1/\eps^{2})$ passes.

For general graphs, whether $\Poly(1/\eps)$ passes suffice remained open for a long time, until Fischer, Mitrovi{\'c}, and Uitto~\cite{fischer2022deterministic} gave a deterministic semi-streaming algorithm with $O(1/\eps^{19})$ passes.
The pass complexity was subsequently improved to $O(1/\eps^{6})$ by Mitrovi\'{c}, Mukherjee, Sankowski, and Sheu~\cite{mitrovic_et_al:LIPIcs.ICALP.2025.119}, and to $O(1/\eps^{4})$ by Izumi, Kitamura, and Yamaguchi~\cite{izumi2025forgetting}.


Assadi~\cite{assadi2025simple} posed, as an open question, whether this gap can be closed, i.e., whether $O(1/\eps^{2})$ passes are achievable for general graphs.
We answer this question affirmatively.

\begin{restatable}{theorem}{main} \label{thm:main}
    For any $\eps > 0$, there exists a deterministic semi-streaming algorithm that computes a $(1 - \eps)$-approximate maximum matching 
    in $O(1/\eps^2)$ passes.
    The algorithm uses $O(n/\eps^2)$ words of space.\footnote{Throughout this paper, a word consists of $O(\log(n/\eps))$ bits.}
\end{restatable}

Our approach, which we call the \emph{budgeted price approach}, is very different from the previous $\mathrm{poly}(1/\varepsilon)$-pass semi-streaming algorithms for the $(1-\varepsilon)$-approximation of matching in general graphs~\cite{fischer2022deterministic,mitrovic_et_al:LIPIcs.ICALP.2025.119,izumi2025forgetting}.
\important{Our algorithm computes a sequence of $O(1/\eps^{2})$ weighted matchings in subgraphs of the input graph, where the edge weights are carefully chosen at each step.}
Moreover, a constant-factor approximate maximum weight matching suffices at each step, which can be computed in a single pass~\cite{feigenbaum2005graph, epstein2011improved, zelke2012weighted, crouch2014improved, paz20182, ghaffari_et_al:OASIcs.SOSA.2019.13}.
Hence, our algorithm can be implemented in the semi-streaming model with $O(1/\eps^{2})$ passes.

\subsection{Additional Related Work}

Semi-streaming algorithms with $O(\log n \cdot \Poly(1/\eps))$ passes have also been studied~\cite{ahn2011laminar,ahn2013linear,ahn2018access,assadi2022semi,assadi2025simple}, and in particular, $O(\log(n)/\eps)$ passes are achievable for general graphs~\cite{ahn2011laminar,ahn2018access,assadi2025simple}.

In the single-pass setting, the greedy algorithm that computes a maximal matching gives a $1/2$-approximation, and a line of work has established lower bounds on the achievable approximation ratio~\cite{goel2012communication,kapralov2013better,kapralov2021space,assadi2026semi}. 
Whether greedy can be beaten had been one of the most well-known open questions in graph streaming.
Very recently, Assadi, Jiang, and Xiang~\cite{assadi2026semib} resolved it by showing that $1/2$ is optimal.


Recently, Konrad, Naidu, Walton, and Wang~\cite{konrad_et_al:LIPIcs.ESA.2026.120} further simplified the $O(1/\eps^2)$-pass semi-streaming algorithm of Assadi, Liu, and Tarjan~\cite{assadi2021auction} and provided a hard instance on which the algorithm requires $\Omega(1/\eps^2)$ iterations.

Assadi and Sundaresan~\cite{assadi2023hidden} proved a conditional $\Omega(\log(1/\eps))$-pass lower bound, even for $(1-\eps)$-approximating the maximum matching size.
As reported in Sundaresan's Ph.D. thesis~~\cite{sundaresan2026new}, Assadi, Konrad, Naidu, and Sundaresan established the same lower bound unconditionally for semi-streaming algorithms that output a $(1-\eps)$-approximate matching.


\section{Preliminaries} \label{sec:preliminaries}





\paragraph{Notation.}
For $y \in \bR^V$ and $S \subseteq V$, let $y(S) := \sum_{v \in S} y_v$, and for an edge $e = \{u, v\}$, let $y(e) := y_u + y_v$.
We also write $\|y\|_2^2 := \sum_{v \in V} y_v^2$.

For $w:F\to\bR$ and $F'\subseteq F$, let $w(F'):=\sum_{e\in F'}w(e)$.

For $a, b \in \bR^V$, the vector $a + b \in \bR^V$ is defined by $(a + b)_v := a_v + b_v$ for $v \in V$.
We write $a \le b$ if $a_v \le b_v$ for every $v \in V$, and $0$ denotes the all-zero vector.

Throughout, we identify an edge set $H \subseteq E$ with the subgraph $(V, H)$ of $G$, and we write $\mu(H)$ for the maximum size of a matching $M \subseteq H$.
In particular, $\mu(G) := \mu(E)$ is the maximum size of a matching in $G$.
For a matching $M$, $V(M)$ denotes the set of vertices covered by $M$.



\paragraph{Tutte-Berge Formula}

We use the following classical min-max formula for the maximum matching size.
Here, we write $H - U$ for the graph obtained from $(V, H)$ by deleting the vertices in $U$ together with the edges incident to them.

\begin{theorem}[see e.g. Theorem~24.1 of \cite{schrijver2003combinatorial}] \label{thm:tutte_berge}
    For every $H \subseteq E$, there exists a set $U \subseteq V$ such that
    \begin{equation*}
        \mu(H) \ =\ |U| + \sum_{S} \left\lfloor \frac{|S|}{2} \right\rfloor ,
    \end{equation*}
    where $S$ ranges over the vertex sets of the connected components of $H - U$.
\end{theorem}

\subsection{One-Pass Streaming algorithm for Weighted Matching}

We use a one-pass semi-streaming algorithm for constant-factor approximate maximum weight matching.
Feigenbaum, Kannan, McGregor, Suri, and Zhang~\cite{feigenbaum2005graph} gave a $6$-approximation, and the approximation ratio has since been improved by several subsequent works~\cite{epstein2011improved,zelke2012weighted,crouch2014improved,paz20182,ghaffari_et_al:OASIcs.SOSA.2019.13}.

\begin{lemma}[From Theorem 2 of \cite{feigenbaum2005graph}] \label{lemma:one_pass_weighted}
    There exists a deterministic semi-streaming algorithm that computes a $6$-approximate maximum weighted matching in one pass. The algorithm uses $O(n)$ words of space.
\end{lemma}

\section{Budgeted Price Approach} \label{sec:algorithm}

In this section, we present our algorithm and prove Theorem~\ref{thm:main}.
We begin by introducing two notions, prices and budgets, on which both the algorithm and its analysis are built.

\subparagraph*{Prices.}
A vector $y \in [0,1]^V$ is called a \emph{price vector}.\footnote{The name is reminiscent of the auction algorithm of Assadi, Liu, and Tarjan~\cite{assadi2021auction}.}
We call $y_v$ as the price of the vertex $v$.
For an edge $e$, we define the \emph{residual} of $e$ under $y$ by
\begin{equation*}
    r_e(y) := \max\{1 - y(e), 0 \} ,
\end{equation*}
and for $F \subseteq E$ we write $r_y(F) := \sum_{e \in F} r_e(y)$.

\subparagraph*{Budgets.}
For a subgraph $H \subseteq E$ and a vector $y \in \bR^V$, we define
\begin{equation*}
    C_H(y) := \max \Bigl\{ |M| + y(V \setminus V(M)) \Bigm| M \subseteq H \text{ is a matching} \Bigr\}.
\end{equation*}
We call $C_H(y)$ the \emph{budget} of $y$ on $H$.

The following properties are immediate from the definition.
\begin{enumerate}
    \item[(a)] For every $y \in \bR^V$,
    \begin{equation*}
        C_H(y) = y(V) + \max \bigl\{ r_y(M) \bigm| M \subseteq H \text{ is a matching} \bigr\}.
    \end{equation*}
    Indeed, for every matching $M \subseteq H$ we have $|M| + y(V \setminus V(M)) = y(V) + \sum_{e \in M} (1 - y(e))$, because $y(V(M)) = \sum_{e \in M} y(e)$.
    The last sum is at most $r_y(M)$, and it equals $r_y(M') = r_y(M)$ for the submatching $M' := \{ e \in M \mid y(e) \le 1 \}$.
    Hence, the maximum of $\sum_{e \in M} (1 - y(e))$ over all matchings $M \subseteq H$ equals the maximum of $r_y(M)$.
    \item[(b)]  For each fixed matching $M \subseteq H$, the quantity $|M| + y(V \setminus V(M))$ is a constant plus the sum of some coordinates of $y$, and hence it is convex and nondecreasing in each coordinate of $y$.
    Since $C_H$ is the maximum of finitely many such quantities, $C_H$ is also convex and nondecreasing in each coordinate.
    Moreover, $C_H(y) \le C_{H'}(y)$ whenever $H \subseteq H' \subseteq E$, since the maximum is taken over a larger family of matchings.
    \item[(c)] $C_H(0) = \mu(H)$, and $C_H(y) \ge \mu(H)$ for every $y \ge 0$.
\end{enumerate}

\subsection{Probe Procedure} \label{sec:probe}

The core of our algorithm is a procedure $\textsc{Probe}(B, \delta)$, which receives a \emph{budget} $B > 0$ and an accuracy parameter $\delta \in (0, 1)$.
The budget $B$ should be thought of as a guess for $\mu(G)$.
It either finds a matching of size more than $(1 - \delta) B$, or certifies that $B > (1 - \delta)\mu(G)$.
Our overall algorithm then performs a binary search for $\mu(G)$ using $\textsc{Probe}$ as the test (see Section~\ref{sec:whole} for details).

The guarantee of $\textsc{Probe}$ is formalized in the following proposition.
\begin{proposition} \label{prop:probe_main}
    Let $B > 0$ and $\delta \in (0, 1)$.
    $\textsc{Probe}(B, \delta)$ uses $T = O(1/\delta^2)$ passes, and the subgraph $H$ stored by $\textsc{Probe}(B, \delta)$ always has $O(n/\delta^2)$ edges.
    Moreover, if $B \le (1 - \delta)\cdot \mu(G)$, then $\textsc{Probe}(B, \delta)$ returns a matching of size more than $(1 - \delta) B$.
\end{proposition}

\subparagraph*{Parameters.}
We set
\begin{equation*}
    \beta := \frac{\delta}{24}, \qquad
    \zeta := \frac{\beta^2 B}{2} , \qquad
    T := \left\lceil \frac{1}{2 \beta^2} \right\rceil.
\end{equation*}


\subparagraph*{Energy and potential.}
For $H \subseteq E$ and $y \in [0,1]^V$, we define the \emph{energy} of $y$ on $H$ by
\begin{equation*}
    E_H(y) := \|y\|_2^2 + \beta \cdot r_y(H) .
\end{equation*}
Let $\mathcal{P}_H := \{ y \in [0,1]^V \mid C_H(y) \le B \}$ be the set of price vectors whose budget on $H$ is at most $B$, and define the \emph{potential} of $H$ by
\begin{equation*}
    \Phi_H := \min_{y \in \mathcal{P}_H} E_H(y).
\end{equation*}

The procedure maintains a subgraph $H \subseteq E$, which is initially empty.
In each iteration, it first checks whether $H$ already contains a matching of size more than $(1 - \delta) B$, and if so, returns a maximum matching of $H$ (line~\ref{line:probe_early}).
Otherwise, it computes a price vector $a \in \mathcal{P}_H$ with $E_H(a) \le \Phi_H + \zeta$ (line~\ref{line:probe_price}).
This step uses only the stored subgraph $H$ and does not access the stream, and it is implemented by exhaustive search on $H$ with $O(n)$ words of working space (Lemma~\ref{lemma:exhaustive} in Section~\ref{sec:exhaustive}).
It then makes one pass over the stream and computes, among the edges outside $H$, a constant-factor approximate maximum weight matching $Q$ with respect to the weights $r_e(a)$ (line~\ref{line:probe_pass}).
The matching $Q$ is added to $H$, and the procedure proceeds to the next iteration.
After $T$ iterations, it returns a maximum matching of $H$.
See Algorithm~\ref{alg:probe} for the pseudocode of $\textsc{Probe}(B, \delta)$.

To prove Proposition~\ref{prop:probe_main}, we first establish Lemmas~\ref{lemma:upper} and \ref{lemma:growth}, and then combine them to complete the proof.

\begin{algorithm}[t]
    \caption{$\textsc{Probe}(B, \delta)$} \label{alg:probe}
    \begin{algorithmic}[1]
        \Require A budget $B > 0$ and an accuracy parameter $\delta \in (0, 1)$.
        \State $H \gets \emptyset$.
        \For{$t = 1, 2, \ldots, T$}
            \If{$\mu(H) > (1 - \delta) B$} \label{line:probe_early}
                \State \Return a maximum matching of $H$.
            \EndIf
            \State Compute a price vector $a \in \mathcal{P}_H$ with $E_H(a) \le \Phi_H + \zeta$. \label{line:probe_price}
            \State In one pass, run the one-pass algorithm of Lemma~\ref{lemma:one_pass_weighted} on the edges $e \in E \setminus H$ with weights $w(e) := r_e(a)$, and let $Q$ be the output matching. \label{line:probe_pass}
            \State $H \gets H \cup Q$.
        \EndFor
        \State \Return a maximum matching of $H$.
    \end{algorithmic}
\end{algorithm}


\begin{lemma} \label{lemma:upper}
    Let $B > 0$, $\delta \in (0, 1)$, and $H \subseteq E$.
    If $\mu(H) \le (1 - \delta) B$, then there exists $y \in \mathcal{P}_H$ with $E_H(y) \le B$.
    In particular, $\Phi_H \le B$.
\end{lemma}

\begin{proof}
    By Theorem~\ref{thm:tutte_berge}, there exists $U \subseteq V$ such that $\mu(H) = |U| + \sum_{S \in \mathcal{C}} \lfloor |S|/2 \rfloor$, where $\mathcal{C}$ denotes the family of the vertex sets of the connected components of $H - U$.
    We call a component $S \in \mathcal{C}$ \emph{small} if $|S|$ is odd and $|S| \le 2/\delta$, and let $\mathcal{S} \subseteq \mathcal{C}$ be the family of small components.
    We define $y \in [0,1]^V$ by
    \begin{equation*}
        y_v :=
        \begin{cases}
            1   & \text{if } v \in U, \\
            0   & \text{if } v \in S \text{ for some } S \in \mathcal{S}, \\
            1/2 & \text{otherwise},
        \end{cases}
    \end{equation*}
    and we write
    \begin{equation*}
        A := y(V) = |U| + \sum_{S \in \mathcal{C} \setminus \mathcal{S}} \frac{|S|}{2}
        \qquad \text{and} \qquad
        D := \sum_{S \in \mathcal{S}} \frac{|S| - 1}{2} .
    \end{equation*}
    We will show that $A + D \le B$, and that both $C_H(y)$ and $E_H(y)$ are at most $A + D$.
    This proves the lemma: the former bound on $C_H(y)$ gives $C_H(y) \le B$, i.e., $y \in \mathcal{P}_H$, and the latter gives $E_H(y) \le B$, which also implies $\Phi_H \le B$ by the definition of $\Phi_H$.

    We first show that $A + D \le B$.
    Partition $\mathcal{C}$ into the family $\mathcal{C}_{\mathrm{even}}$ of the even components, the family $\mathcal{S}$ of the small components, and the family $\mathcal{L}$ of the remaining components, which are odd with $|S| > 2/\delta$.
    Since $\lfloor |S|/2 \rfloor = |S|/2$ if $|S|$ is even and $\lfloor |S|/2 \rfloor = (|S| - 1)/2$ if $|S|$ is odd, we have
    \begin{align*}
        \mu(H) &= |U| + \sum_{S \in \mathcal{C}_{\mathrm{even}}} \frac{|S|}{2} + \sum_{S \in \mathcal{L}} \frac{|S| - 1}{2} + \sum_{S \in \mathcal{S}} \frac{|S| - 1}{2} , \\
        A + D &= |U| + \sum_{S \in \mathcal{C}_{\mathrm{even}}} \frac{|S|}{2} + \sum_{S \in \mathcal{L}} \frac{|S|}{2} + \sum_{S \in \mathcal{S}} \frac{|S| - 1}{2} ,
    \end{align*}
    and hence $A + D - \mu(H) = |\mathcal{L}| / 2$.
    Every $S \in \mathcal{L}$ satisfies $|S| - 1 > 2/\delta - 1 \ge 1/\delta$ because $\delta \le 1$, and hence $1/2 < \delta \cdot (|S| - 1)/2$.
    Summing this over $S \in \mathcal{L}$, we obtain
    \begin{equation*}
        A + D - \mu(H) \ =\ \sum_{S \in \mathcal{L}} \frac{1}{2} \ \le\ \delta \sum_{S \in \mathcal{L}} \frac{|S| - 1}{2} \ \le\ \delta \cdot \mu(H) .
    \end{equation*}
    Therefore, $A + D \le (1 + \delta) \cdot \mu(H) \le (1 + \delta)(1 - \delta) B \le B$.

    Next, we compute the residual $r_y(e)$ of each edge $e \in H$ under the price vector $y$ defined above.
    These values are used to upper-bound $C_H(y)$ and $E_H(y)$.
    Every edge $e \in H$ either has an endpoint in $U$ or lies inside a component in $\mathcal{C}$.
    In the former case $y(e) \ge 1$ because $y_u = 1$ for $u \in U$, and inside a component that is not small $y(e) = 1/2 + 1/2 = 1$.
    Thus, $r_y(e) = 1$ if $e$ lies inside a small component, where $y(e) = 0$, and $r_y(e) = 0$ otherwise.
    In particular, for every $F \subseteq H$, $r_y(F)$ is the number of edges of $F$ inside small components.

    We now show that $C_H(y) \le A + D$, which implies $y \in \mathcal{P}_H$ since $A + D \le B$.
    By property~(a) of budgets, $C_H(y) = A + \max_M r_y(M)$, where the maximum is taken over all matchings $M \subseteq H$.
    Fix a matching $M \subseteq H$.
    Here, $r_y(M)$ is the number of edges of $M$ inside small components, that is,
    \begin{equation*}
        r_y(M) \ =\ \sum_{S \in \mathcal{S}} \bigl| \{ e \in M \mid e \subseteq S \} \bigr| .
    \end{equation*}
    Since the edges of $M$ inside $S$ are pairwise disjoint and $|S|$ is odd, there are at most $(|S| - 1)/2$ of them, and hence $r_y(M) \le \sum_{S \in \mathcal{S}} (|S| - 1)/2 = D$.
    Therefore, $C_H(y) \le A + D$.

    Finally, we show that $E_H(y) \le A + D$, where $E_H(y) = \|y\|_2^2 + \beta \, r_y(H)$.
    For the first term, since $y_v \in [0, 1]$ implies $y_v^2 \le y_v$ for every $v \in V$, we have $\|y\|_2^2 \le y(V) = A$.
    For the second term, recall that $r_y(H)$ is the number of edges of $H$ inside small components, i.e.,
    \begin{equation*}
        r_y(H) \ =\ \sum_{S \in \mathcal{S}} \bigl| \{ e \in H \mid e \subseteq S \} \bigr| .
    \end{equation*}
    Here, the number of edges inside $S$ is at most $\binom{|S|}{2} = |S| \cdot \frac{|S| - 1}{2}$.
    Moreover, every small component satisfies $|S| \le 2/\delta$, and hence $\beta |S| \le 2\beta/\delta \le 1$ since $\beta = \delta/24$.
    Thus,
    \begin{equation*}
        \beta \, r_y(H) \ \le\ \sum_{S \in \mathcal{S}} \beta \binom{|S|}{2} \ =\ \sum_{S \in \mathcal{S}} \beta |S| \cdot \frac{|S| - 1}{2} \ \le\ \sum_{S \in \mathcal{S}} \frac{|S| - 1}{2} \ =\ D .
    \end{equation*}
    Therefore, $E_H(y) \le A + D \le B$, which completes the proof.
\end{proof}


\begin{lemma} \label{lemma:growth}
    Let $B > 0$, $\delta \in (0, 1)$, and $H \subseteq E$.
    Let $a \in \mathcal{P}_H$ be a price vector with $E_H(a) \le \Phi_H + \zeta$, and let $Q \subseteq E \setminus H$ be a matching such that $\mathcal{P}_{H \cup Q} \neq \emptyset$.
    Then, we have
    \begin{equation*}
        \Phi_{H \cup Q} \ \ge\ \Phi_H + \beta \, r_a(Q) - \beta^2 |Q| - 2 \zeta .
    \end{equation*}
\end{lemma}

\begin{proof}
    We first show the following claim.

    \begin{claim} \label{claim:stable}
        For every $b \in \mathcal{P}_H$, we have $E_H(b) \ge E_H(a) + \frac{1}{2} \|b - a\|_2^2 - 2 \zeta$.
    \end{claim}

    \begin{proof}[Proof of Claim~\ref{claim:stable}]
        Let $m := (a + b)/2$.
        Then, $m \in \mathcal{P}_H$, because $\mathcal{P}_H$ is convex.
        Since each $r_e$ is convex, the function $E_H(y) - \|y\|_2^2 = \beta \cdot r_y(H)$ is convex, and we have $\|m\|_2^2 = \frac{1}{2}\|a\|_2^2 + \frac{1}{2}\|b\|_2^2 - \frac{1}{4}\|b - a\|_2^2$.
        Thus,
        \begin{align*}
            E_H(a) - \zeta \ \le\ \Phi_H \ \le\ E_H(m)
            &=\ \|m\|_2^2 + \beta \, r_m(H) \\
            &=\ \frac{1}{2} \|a\|_2^2 + \frac{1}{2} \|b\|_2^2 - \frac{1}{4} \|b - a\|_2^2 + \beta \, r_m(H) \\
            &\le\ \frac{1}{2} \|a\|_2^2 + \frac{1}{2} \|b\|_2^2 - \frac{1}{4} \|b - a\|_2^2 + \frac{\beta}{2} \, r_a(H) + \frac{\beta}{2} \, r_b(H) \\
            &=\ \frac{1}{2} E_H(a) + \frac{1}{2} E_H(b) - \frac{1}{4} \|b - a\|_2^2 ,
        \end{align*}
        where the third inequality holds by the convexity of $y \mapsto r_y(e)$ for each $e \in H$.
        Rearranging the terms completes the proof of the claim.
    \end{proof}

    Let $b \in \mathcal{P}_{H \cup Q}$ be a price vector attaining $\Phi_{H \cup Q}$.
    Since $Q \cap H = \emptyset$, we have $r_b(H \cup Q) = r_b(H) + r_b(Q)$, and hence
    \begin{equation*}
        \Phi_{H \cup Q} \ =\ E_{H \cup Q}(b) = \|b\|_2^2 + \beta \cdot (r_b(H) + r_b(Q)) =\ E_H(b) + \beta \, r_b(Q) .
    \end{equation*}
    Since $C_H \le C_{H \cup Q}$, we have $\mathcal{P}_{H \cup Q} \subseteq \mathcal{P}_H$, and hence Claim~\ref{claim:stable} applies to $b \in \mathcal{P}_H$, which gives
    \begin{equation*}
        \Phi_{H \cup Q} = E_H(b) + \beta r_b(Q) \ge\ E_H(a) + \frac{1}{2} \|b - a\|_2^2 - 2\zeta + \beta \, r_b(Q) .
    \end{equation*}

    Since $E_H(a) \ge \Phi_H$, it remains to show that
    \begin{equation*}
        \frac{1}{2} \|b - a\|_2^2 + \beta \, r_b(Q) \ \ge\ \beta \, r_a(Q) - \beta^2 |Q| .
    \end{equation*}
    Write $\Delta_v := b_v - a_v$ for $v \in V$.
    Since $Q$ is a matching, each vertex of $V(Q)$ belongs to exactly one edge of $Q$, and hence
    \begin{equation*}
        \|b - a\|_2^2 \ =\ \sum_{v \in V} \Delta_v^2 \ \ge\ \sum_{v \in V(Q)} \Delta_v^2 \ =\ \sum_{\{u, v\} \in Q} \left( \Delta_u^2 + \Delta_v^2 \right) .
    \end{equation*}
    Therefore, it suffices to show that every edge $e = \{u, v\} \in Q$ satisfies
    \begin{equation*}
        \frac{1}{2} \left( \Delta_u^2 + \Delta_v^2 \right) + \beta \, r_b(e) \ \ge\ \beta \, r_a(e) - \beta^2 .
    \end{equation*}
    
    If $r_e(a) = 0$, this holds because the left-hand side is nonnegative.
    Otherwise, $r_e(a) = 1 - a(e)$, and $r_b(e) \ =\ \max\{ 1 - b(e), 0 \} \ \ge\ 1 - b(e) \ =\ 1 - a(e) - \Delta_u - \Delta_v \ =\ r_a(e) - \Delta_u - \Delta_v$.
    Hence,
    \begin{align*}
        \frac{1}{2} \left( \Delta_u^2 + \Delta_v^2 \right) + \beta \, r_b(e)
        &\ge\ \frac{1}{2} \Delta_u^2 + \frac{1}{2} \Delta_v^2 + \beta \, r_a(e) - \beta \Delta_u - \beta \Delta_v \\
        &=\ \beta \, r_a(e) + \left( \frac{1}{2} \Delta_u^2 - \beta \Delta_u \right) + \left( \frac{1}{2} \Delta_v^2 - \beta \Delta_v \right) \\
        &\ge\ \beta \, r_a(e) - \frac{1}{2} \beta^2 - \frac{1}{2} \beta^2
        \ =\ \beta \, r_a(e) - \beta^2 ,
    \end{align*}
    where the last inequality holds because $\frac{1}{2} \Delta^2 - \beta \Delta = \frac{1}{2} (\Delta - \beta)^2 - \frac{1}{2} \beta^2 \ge -\frac{1}{2} \beta^2$ for every real $\Delta$.
    This completes the proof.
\end{proof}

We are now ready to prove Proposition~\ref{prop:probe_main}.

\begin{proof}[Proof of Proposition~\ref{prop:probe_main}]
    
    Suppose that $B \le (1 - \delta)\mu(G)$.
    Suppose, for a contradiction, that $\textsc{Probe}(B, \delta)$ returns a matching of size at most $(1 - \delta) \cdot B$.
    The returned matching is a maximum matching of the final subgraph $H$, and $H$ only grows during the execution.
    Hence, $\mu(H) \le (1 - \delta) B$ holds at every moment of the execution.
    For $t = 1, 2, \ldots, T + 1$, let $H_t$ be the subgraph $H$ at the beginning of the $t$-th iteration,\footnote{For $t = T + 1$, this means the subgraph $H$ at the end of the $T$-th iteration.} so that $H_1 = \emptyset$ and $H_{t+1} = H_t \cup Q_t$, where $a_t$ and $Q_t$ denote the price vector $a$ and the matching $Q$ in the $t$-th iteration.
    By Lemma~\ref{lemma:upper}, we have $\mathcal{P}_{H_t} \neq \emptyset$ and $\Phi_{H_t} \le B$ for every $t = 1, 2, \ldots, T + 1$. 

    Now, we show that the potential grows by at least $2 \beta^2 \cdot \mu(G)$ in every iteration.
    To lower-bound $r_{a_t}(Q_t)$, we first show the following claim.

    \begin{claim} \label{claim:unseen}
        For every $H \subseteq E$, every $y \in [0,1]^V$, and every matching $N \subseteq E$, we have $r_y(N \setminus H) \ge |N| - C_H(y)$.
    \end{claim}

    \begin{proof}[Proof of Claim~\ref{claim:unseen}]
        Apply the definition of $C_H(y)$ to the matching $M := N \cap H$ of $H$, which gives $C_H(y) \ge |N \cap H| + y(V \setminus V(M))$.
        Since $N$ is a matching, the endpoints of the edges of $N \setminus H$ are distinct and are not covered by $M$, and hence $y(V \setminus V(M)) \ge \sum_{e \in N \setminus H} y(e)$ because $y \ge 0$.
        Therefore,
        \begin{equation*}
            C_H(y) \ \ge\ |N \cap H| + \sum_{e \in N \setminus H} y(e) \ =\ |N| - \sum_{e \in N \setminus H} \bigl( 1 - y(e) \bigr) \ \ge\ |N| - r_y(N \setminus H) ,
        \end{equation*}
        where the last inequality holds because $1 - y(e) \le r_y(e)$.
    \end{proof}

    Fix $t \in \{1, 2, \ldots, T\}$ and let $N$ be a maximum matching of $G$.
    Since $a_t \in \mathcal{P}_{H_t}$, we have $C_{H_t}(a_t) \le B$, and Claim~\ref{claim:unseen} yields
    \begin{equation*}
        r_{a_t}(N \setminus H_t) \ge |N| - C_{H_t}(a_t) \ge \mu(G) - B \ \ge\ \delta \cdot \mu(G) .
    \end{equation*}
    Since $N \setminus H_t$ is a matching consisting of edges in $E \setminus H_t$, and the matching $Q_t$ found in Line~\ref{line:probe_pass} is a $6$-approximate maximum weight matching of $E \setminus H_t$ with respect to the weights $r_{a_t}$ by Lemma~\ref{lemma:one_pass_weighted}, we have
    \begin{equation*}
        r_{a_t}(Q_t) \ \ge\ \frac{1}{6} \, r_{a_t}(N \setminus H_t) \ \ge\ \frac{\delta \cdot \mu(G)}{6} \ =\ 4 \beta \cdot \mu(G),
    \end{equation*}
    where the last equality holds because $\beta = \delta / 24$.
    We apply Lemma~\ref{lemma:growth} to $H_t$, $a_t$, and $Q_t$.
    Now, $\mathcal{P}_{H_{t+1}} \neq \emptyset$ holds as shown above.
    Using $|Q_t| \le \mu(G)$ and $2\zeta = \beta^2 B \le \beta^2 \cdot \mu(G)$, we obtain
    \begin{equation*}
        \Phi_{H_{t+1}} - \Phi_{H_t}
        \ \ge\ \beta \, r_{a_t}(Q_t) - \beta^2 |Q_t| - 2 \zeta
        \ \ge\ 4 \beta^2 \cdot \mu(G) - \beta^2 \cdot \mu(G) - \beta^2 \cdot \mu(G)
        \ =\ 2 \beta^2 \cdot \mu(G) .
    \end{equation*}

    Summing the above inequality over $t = 1, 2, \ldots, T$ and using $\Phi_{H_1} = \Phi_{\emptyset} = 0$, we obtain
    \begin{equation*}
        \Phi_{H_{T + 1}} \ \ge\ 2 \beta^2 T \cdot \mu(G) \ \ge\ \mu(G) \ >\ B ,
    \end{equation*}
    where the second inequality holds because $T \ge 1 / (2 \beta^2)$, and the last inequality holds because $B \le (1 - \delta) \cdot \mu(G) < \mu(G)$.
    This contradicts $\Phi_{H_{T+1}} \le B$, which completes the proof.
\end{proof}

\subsection{Whole Algorithm} \label{sec:whole}


In this section, we prove Theorem~\ref{thm:main}.
The overall algorithm performs a binary search for $\mu(G)$, using $\textsc{Probe}$ (Algorithm~\ref{alg:probe}) as the test.
See Algorithm~\ref{alg:whole} for the pseudocode.

\begin{algorithm}[t]
    \caption{Whole algorithm} \label{alg:whole}
    \begin{algorithmic}[1]
        \State In one pass, compute a maximal matching $M_0$ of $G$ greedily, and let $s := |M_0|$.
        \State $L \gets s$, $U \gets 2s$, and $N^* \gets M_0$.
        \While{$U - L > \eps s$} 
            \State $B \gets (L + U) / 2$. 
            \State $w \gets U - L$ and $\ds \delta \gets \frac{w}{10s}$.
            \State $N \gets \textsc{Probe}(B, \delta)$, and discard the subgraph stored by $\textsc{Probe}$. \label{line:whole_probe}
            \If{$|N| > (1 - \delta) B$}
                \State $L \gets |N|$ and $N^* \gets N$. 
            \Else
                \State $U \gets \lfloor B / (1 - \delta) \rfloor$. 
            \EndIf
        \EndWhile
        \State \Return $N^*$.
    \end{algorithmic}
\end{algorithm}

\main*

\begin{proof}[Proof of Theorem~\ref{thm:main}]
    We prove that Algorithm~\ref{alg:whole} returns a matching of size at least $(1 - \eps)\cdot\mu(G)$ in $O(1/\eps^2)$ passes using $O(n/\eps^2)$ words of space.
    We first prove the following claim.

    \begin{claim} \label{claim:invariant}
        Throughout the execution of Algorithm~\ref{alg:whole}, $N^*$ is a matching of $G$ with $|N^*| = L$, and $s \le L \le \mu(G) \le U \le 2s$.
        Moreover, each iteration of the while loop replaces the width $w = U - L$ by a width at most $\frac{9}{10} w$.
    \end{claim}
    
    \begin{proof}[Proof of Claim~\ref{claim:invariant}]
        We proceed by induction on the number of iterations.
        Initially, $|N^*| = |M_0| = s = L$, and $s \le \mu \le 2s$ holds because a maximal matching has at least half the size of a maximum matching.
        Consider an iteration that starts with $L$ and $U$ satisfying $s \le L \le \mu(G) \le U \le 2s$, and let $w = U - L$.
        Let $L'$ and $U'$ denote the values of $L$ and $U$ after the iteration, respectively.
        Then, $w \le s$, and hence $\delta = w / (10 s) \le 1/10$.
        Since $B = L + \frac{w}{2} \le U \le 2s$, we have
        \begin{equation*}
            \delta B \ \le\ \delta U \ \le\ \frac{2 w}{10} \ =\ \frac{w}{5} .
        \end{equation*}
    
        Suppose that $|N| > (1 - \delta) B$ holds in the iteration.
        Then, the algorithm updates $L$ to $|N|$ and $N^*$ to $N$, so $L' = |N|$ and $U' = U$.
        Since $N$ is a matching of $G$, we have $L' \le \mu(G)$.
        Moreover, $L' = |N| > B - \delta B \ge L + \frac{w}{2} - \frac{w}{5}$, and hence the new width is $U' - L' < \frac{7}{10} w$.
        
        Suppose instead that $|N| \le (1 - \delta) B$ holds in the iteration.
        Then, by Proposition~\ref{prop:probe_main}, we have $B > (1 - \delta)\cdot\mu(G)$, i.e., $\mu(G) < B / (1 - \delta)$.
        The algorithm updates $U$ to $\lfloor B / (1 - \delta) \rfloor$, so $L' = L$ and $U' = \lfloor B / (1 - \delta) \rfloor$, and we have $\mu(G) \le U'$.
        Moreover, since $1 / (1 - \delta) \le 1 + 2\delta$ for $\delta \le 1/2$, we have $U' \le B + 2 \delta B \le L' + \frac{w}{2} + \frac{2 w}{5}$, and hence the new width is $U' - L' \le \frac{9}{10} w$.
        
        In both cases, the new interval $[L', U']$ is contained in the old one $[L, U]$ and still contains $\mu(G)$, and $|N^*| = L$ is maintained.
        This completes the induction and the proof of the claim.
    \end{proof}

    \subparagraph*{Approximation.}
    The while loop terminates because $U - L$ strictly decreases in every iteration.
    When it terminates, $U - L \le \eps s$, and hence, by Claim~\ref{claim:invariant}, the output satisfies $|N^*| = L \ge U - \eps s \ge \mu(G) - \eps \mu(G)$.

    \subparagraph*{Passes.}
    Let $w_0 > w_1 > \cdots > w_{k-1}$ be the widths at the beginnings of the $k$ iterations of the while loop.
    By Claim~\ref{claim:invariant}, we have $w_i \ge (10/9)^{k-1-i} w_{k-1}$ for every $i$, and the condition of the while loop gives $w_{k-1} > \eps s$, where we note that $w_0 = s$.
    By Proposition~\ref{prop:probe_main}, the $i$-th iteration makes $O(1/\delta_i^2) = O(s^2 / w_i^2)$ passes, where $\delta_i :=  \frac{w_i}{10 s}$.
    Therefore, the total number of passes is
    \begin{equation*}
        1 + \sum_{i = 0}^{k-1} O\!\left( \frac{s^2}{w_i^2} \right)
        \ =\ O\!\left( \frac{s^2}{w_{k-1}^2} \sum_{i = 0}^{k-1} \left( \frac{81}{100} \right)^{k - 1 - i} \right)
        \ =\ O\!\left( \frac{s^2}{w_{k-1}^2} \right)
        \ =\ O\!\left( \frac{1}{\eps^2} \right) .
    \end{equation*}

    \subparagraph*{Space.}
    Consider the $i$-th iteration.
    Since $w_i > \eps s$, we have $\delta_i > \frac{\eps}{10}$, and hence, by Proposition~\ref{prop:probe_main}, the subgraph stored by $\textsc{Probe}(B, \delta_i)$ has $O(n / \delta_i^2) = O(n/\eps^2)$ edges.
    Note that it is discarded before the next iteration.
    Besides the stored subgraph, the algorithm keeps the matchings $N^*$, $N$, and $Q$, the price vector $a$, which take $O(n)$ words in total.
    \tatsuya{By Lemma~\ref{lemma:exhaustive}, every number that the algorithm stores fits in one word of $O(\log(n/\eps))$ bits.}
    Therefore, the algorithm uses $O(n/\eps^2)$ words of space, which completes the proof.
\end{proof}

\subsection{Computing Prices by Exhaustive Search} \label{sec:exhaustive}

It remains to show how Line~\ref{line:probe_price} of $\textsc{Probe}$ is implemented in small space.
We implement these steps by exhaustive search. 
Here, all computation is done on the stored subgraph $H$ without accessing the stream.

\begin{lemma} \label{lemma:exhaustive}
    Let $H \subseteq E$ with $\mu(H) \le B$.
    Then, a price vector $a \in \mathcal{P}_H$ with $E_H(a) \le \Phi_H + \zeta$ can be computed from $H$ using $O(n)$ words of working space.
\end{lemma}

\begin{proof}
    Let $m := |H|$ and $W := \max \{ \lceil 2 \beta m / \zeta \rceil, 1 \}$, and let $\Gamma_W := \{ 0, 1/W, 2/W, \ldots, 1 \}^V$.
    The algorithm enumerates all points $y \in \Gamma_W$, and among those with $C_H(y) \le B$, returns a point $a$ minimizing $E_H(y)$.

    Let $x \in \mathcal{P}_H$ be a price vector attaining $\Phi_H$, and define $z \in [0,1]^V$ by $z_v := \lfloor W x_v \rfloor / W$ for $v \in V$.
    Then, $z \in \Gamma_W$ and $0 \le z \le x$, and hence $C_H(z) \le C_H(x) \le B$ because $C_H$ is nondecreasing in each coordinate (property~(b)).
    Thus, $z$ satisfies the constraint, so $a$ exists, and $a \in \mathcal{P}_H$ because $a \in \Gamma_W \subseteq [0,1]^V$ and $C_H(a) \le B$.
    By the minimality of $a$, we have $E_H(a) \le E_H(z)$.
    Moreover, $\|z\|_2^2 \le \|x\|_2^2$ since $0 \le z \le x$, and $z(e) \ge x(e) - 2/W$ implies $r_z(e) \le r_x(e) + 2/W$ for every edge $e$.
    Therefore,
    \begin{equation*}
        E_H(a) \ \le\ E_H(z) \ \le\ \|x\|_2^2 + \beta \left( r_x(H) + \frac{2m}{W} \right) \ =\ \Phi_H + \frac{2 \beta m}{W} \ \le\ \Phi_H + \zeta ,
    \end{equation*}
    where the last inequality holds by the choice of $W$.
\end{proof}

\section{Further Results} \label{sec:further}

Our approach extends in several directions.
We only state the results here and defer the details to a forthcoming manuscript.
\begin{itemize}
    \item \textbf{Weighted matching.}
    Our approach yields a semi-streaming algorithm that computes a $(1 - \eps)$-approximate maximum weight matching in $O(\log(1/\eps) / \eps^2)$ passes.
    \item \textbf{Massively Parallel Computation.}
    Our approach also works in the Massively Parallel Computation (MPC) model.
    \item \textbf{Running time.}
    The running time of the algorithm in this paper is exponential in $n$ and $1/\eps$ because of the exhaustive search in Section~\ref{sec:exhaustive}.
    It can be made polynomial in $n$ and $1/\eps$ while keeping the same pass and space bounds.
\end{itemize}

\paragraph*{Acknowledgment}
A manuscript containing the results in Section~\ref{sec:further}, together with their full proofs, will be made publicly available shortly.

\paragraph*{Declaration of Generative AI Use}
The main ideas behind our budgeted price approach and its analysis were developed largely with the assistance of GPT-6 Astra (OpenAI).
We also used GPT-6 Astra (OpenAI) and Claude Fable 5.1 (Anthropic) to assist in drafting the manuscript and improving its language and presentation.
We carefully verified and revised the proofs and take full responsibility for the paper.

\bibliography{biblio}

@InProceedings{mitrovic_et_al:LIPIcs.ICALP.2025.119,
  author =	{Mitrovi\'{c}, Slobodan and Mukherjee, Anish and Sankowski, Piotr and Sheu, Wen-Horng},
  title =	{Faster Semi-Streaming Matchings via Alternating Trees},
  booktitle =	{Proceedings of  the 52nd International Colloquium on Automata, Languages, and Programming (ICALP 2025)},
  pages =	{119:1--119:19},
  year =	{2025},
  volume =	{334},
  doi =		{10.4230/LIPIcs.ICALP.2025.119},
}

@article{izumi2025forgetting,
  title={Forgetting Alternation and Blossoms: A New Framework for Fast Matching Augmentation and Its Applications to Sequential/Distributed/Streaming Computation},
  author={Izumi, Taisuke and Kitamura, Naoki and Yamaguchi, Yutaro},
  journal={arXiv preprint arXiv:2511.08210},
  year={2025}
}

@inproceedings{fischer2022deterministic,
  title={Deterministic {$(1+ \varepsilon)$}-approximate maximum matching with poly {$(1/\varepsilon)$} passes in the semi-streaming model and beyond},
  author={Fischer, Manuela and Mitrovi{\'c}, Slobodan and Uitto, Jara},
  booktitle={Proceedings of  of the 54th Annual ACM SIGACT Symposium on Theory of Computing (STOC 2022)},
  pages={248--260},
  year={2022}
}

@article{assadi2025simple,
  title={A Simple {$(1-\varepsilon)$}-Approximation Semi-Streaming Algorithm for Maximum (Weighted) Matching},
  author={Assadi, Sepehr},
  journal={TheoretiCS},
  volume={4},
  year={2025},
  publisher={Episciences. org},
  note={Announced at SOSA'24.}
}

@article{feigenbaum2005graph,
  title = {On graph problems in a semi-streaming model},
  author = {Feigenbaum, Joan and Kannan, Sampath and McGregor, Andrew and Suri, Siddharth and Zhang, Jian},
  journal = {Theoretical Computer Science},
  volume = {348},
  number = {2-3},
  pages = {207--216},
  year = {2005},
  publisher = {Elsevier},
  note = {Announced at ICALP 2004.},
  timestamp = {Wed, 17 Feb 2021 00:00:00 +0100},
  biburl = {https://dblp.org/rec/journals/tcs/FeigenbaumKMSZ05.bib},
  bibsource = {dblp computer science bibliography, https://dblp.org},
  doi = {10.1016/j.tcs.2005.09.013},
  _bib2doi_selected = {dblp:/rec/journals/tcs/FeigenbaumKMSZ05.bib},
  _bib2doi_confirmed = {true},
}

@inproceedings{mcgregor2005finding,
  title = {Finding graph matchings in data streams},
  author = {McGregor, Andrew},
  booktitle = {Proceedings of the 8th International Workshop on Approximation Algorithms for Combinatorial Optimization Problems (APPROX 2005) and the 9th International Workshop on Randomized and Computation (RANDOM 2005)},
  pages = {170--181},
  year = {2005},
  organization = {Springer},
  timestamp = {Tue, 23 May 2017 01:00:00 +0200},
  biburl = {https://dblp.org/rec/conf/approx/McGregor05.bib},
  bibsource = {dblp computer science bibliography, https://dblp.org},
  doi = {10.1007/11538462_15},
  _bib2doi_selected = {dblp:/rec/conf/approx/McGregor05.bib},
  _bib2doi_confirmed = {true},
}

@inproceedings{eggert2009bipartite,
  title = {Bipartite graph matchings in the semi-streaming model},
  author = {Eggert, Sebastian and Kliemann, Lasse and Srivastav, Anand},
  booktitle = {Proceedings of the 17th Annual European Symposium on Algorithms (ESA 2009)},
  pages = {492--503},
  year = {2009},
  organization = {Springer},
  timestamp = {Fri, 26 May 2017 01:00:00 +0200},
  biburl = {https://dblp.org/rec/conf/esa/EggertKS09.bib},
  bibsource = {dblp computer science bibliography, https://dblp.org},
  doi = {10.1007/978-3-642-04128-0_44},
  _bib2doi_selected = {dblp:/rec/conf/esa/EggertKS09.bib},
  _bib2doi_confirmed = {true},
}

@article{ahn2013linear,
  title = {Linear programming in the semi-streaming model with application to the maximum matching problem},
  author = {Ahn, Kook Jin and Guha, Sudipto},
  journal = {Information and Computation},
  volume = {222},
  pages = {59--79},
  year = {2013},
  publisher = {Elsevier},
  note = {Announced at ICALP 2011.},
  timestamp = {Fri, 12 Feb 2021 00:00:00 +0100},
  biburl = {https://dblp.org/rec/journals/iandc/AhnG13.bib},
  bibsource = {dblp computer science bibliography, https://dblp.org},
  doi = {10.1016/j.ic.2012.10.006},
  _bib2doi_selected = {dblp:/rec/journals/iandc/AhnG13.bib},
  _bib2doi_confirmed = {true},
}

@article{eggert2012bipartite,
  title = {Bipartite matching in the semi-streaming model},
  author = {Eggert, Sebastian and Kliemann, Lasse and Munstermann, Peter and Srivastav, Anand},
  journal = {Algorithmica},
  volume = {63},
  number = {1},
  pages = {490--508},
  year = {2012},
  publisher = {Springer},
  timestamp = {Wed, 17 May 2017 01:00:00 +0200},
  biburl = {https://dblp.org/rec/journals/algorithmica/EggertKMS12.bib},
  bibsource = {dblp computer science bibliography, https://dblp.org},
  doi = {10.1007/s00453-011-9556-8},
  _bib2doi_selected = {dblp:/rec/journals/algorithmica/EggertKMS12.bib},
  _bib2doi_confirmed = {true},
}

@inproceedings{kapralov2013better,
  title = {Better bounds for matchings in the streaming model},
  author = {Kapralov, Michael},
  booktitle = {Proceedings of the 24th Annual ACM-SIAM Symposium on Discrete Algorithms (SODA 2013)},
  pages = {1679--1697},
  year = {2013},
  organization = {SIAM},
  timestamp = {Sun, 06 Oct 2024 01:00:00 +0200},
  biburl = {https://dblp.org/rec/conf/soda/Kapralov13.bib},
  bibsource = {dblp computer science bibliography, https://dblp.org},
  doi = {10.1137/1.9781611973105.121},
  _bib2doi_selected = {dblp:/rec/conf/soda/Kapralov13.bib},
  _bib2doi_confirmed = {true},
}

@inproceedings{kapralov2021space,
  title={Space lower bounds for approximating maximum matching in the edge arrival model},
  author={Kapralov, Michael},
  booktitle={Proceedings of the 32th ACM-SIAM Symposium on Discrete Algorithms (SODA 2021)},
  pages={1874--1893},
  year={2021},
  organization={SIAM}
}

@inproceedings{assadi2026semi,
  title={Semi-streaming matching in a single pass: A new framework for lower bounds via blueprints},
  author={Assadi, Sepehr and Jiang, Max and Xiang, Mars},
  booktitle={Proceedings of the 58th Annual ACM Symposium on Theory of Computing (STOC 2026)},
  pages={979--990},
  year={2026}
}

@inproceedings{tirodkar2018deterministic,
  title = {Deterministic algorithms for maximum matching on general graphs in the semi-streaming model},
  author = {Tirodkar, Sumedh},
  booktitle = {Proceedings of the 38th IARCS Annual Conference on Foundations of Software Technology and Theoretical Computer Science (FSTTCS 2018)},
  year = {2018},
  organization = {Schloss-Dagstuhl-Leibniz Zentrum f{\"u}r Informatik},
  doi={10.4230/LIPIcs.FSTTCS.2018.39}
}

@inproceedings{assadi2021auction,
  title = {An auction algorithm for bipartite matching in streaming and massively parallel computation models},
  author = {Assadi, Sepehr and Liu, S Cliff and Tarjan, Robert E},
  booktitle = {Proceedings of the 4th Symposium on Simplicity in Algorithms (SOSA 2021)},
  pages = {165--171},
  year = {2021},
  organization = {SIAM},
  timestamp = {Wed, 17 Mar 2021 00:00:00 +0100},
  biburl = {https://dblp.org/rec/conf/sosa/AssadiLT21.bib},
  bibsource = {dblp computer science bibliography, https://dblp.org},
  doi = {10.1137/1.9781611976496.18},
  _bib2doi_selected = {dblp:/rec/conf/sosa/AssadiLT21.bib},
  _bib2doi_confirmed = {true},
}

@article{ahn2011laminar,
  title = {Laminar families and metric embeddings: {Non-bipartite} maximum matching problem in the semi-streaming model},
  author = {Ahn, Kook Jin and Guha, Sudipto},
  journal = {arXiv preprint arXiv:1104.4058},
  year = {2011},
  doi={10.48550/arXiv.1104.4058}
}

@article{ahn2018access,
  title = {Access to data and number of iterations: {Dual} primal algorithms for maximum matching under resource constraints},
  author = {Ahn, Kook Jin and Guha, Sudipto},
  journal = {ACM Transactions on Parallel Computing (TOPC)},
  volume = {4},
  number = {4},
  pages = {1--40},
  year = {2018},
  publisher = {ACM},
  note = {Announced at SPAA 2015.},
  timestamp = {Sat, 25 Apr 2020 01:00:00 +0200},
  biburl = {https://dblp.org/rec/journals/topc/AhnG18.bib},
  bibsource = {dblp computer science bibliography, https://dblp.org},
  doi = {10.1145/3154855},
  _bib2doi_selected = {dblp:/rec/journals/topc/AhnG18.bib},
  _bib2doi_confirmed = {true},
}

@inproceedings{assadi2022semi,
  title = {Semi-Streaming Bipartite Matching in Fewer Passes and Optimal Space},
  author = {Assadi, Sepehr and Jambulapati, Arun and Jin, Yujia and Sidford, Aaron and Tian, Kevin},
  booktitle = {Proceedings of the 33rd Annual ACM-SIAM Symposium on Discrete Algorithms (SODA 2022)},
  pages = {627--669},
  year = {2022},
  organization = {SIAM},
  timestamp = {Tue, 12 Apr 2022 01:00:00 +0200},
  biburl = {https://dblp.org/rec/conf/soda/AssadiJJST22.bib},
  bibsource = {dblp computer science bibliography, https://dblp.org},
  doi = {10.1137/1.9781611977073.29},
  _bib2doi_selected = {dblp:/rec/conf/soda/AssadiJJST22.bib},
  _bib2doi_confirmed = {true},
}

@inproceedings{goel2012communication,
  title = {On the communication and streaming complexity of maximum bipartite matching},
  author = {Goel, Ashish and Kapralov, Michael and Khanna, Sanjeev},
  booktitle = {Proceedings of the 23rd Annual ACM-SIAM Symposium on Discrete Algorithms (SODA 2012)},
  pages = {468--485},
  year = {2012},
  organization = {SIAM},
  timestamp = {Sun, 06 Oct 2024 01:00:00 +0200},
  biburl = {https://dblp.org/rec/conf/soda/GoelKK12.bib},
  bibsource = {dblp computer science bibliography, https://dblp.org},
  doi = {10.1137/1.9781611973099.41},
  _bib2doi_selected = {dblp:/rec/conf/soda/GoelKK12.bib},
  _bib2doi_confirmed = {true},
}

@article{assadi2026semib,
  title={Semi-Streaming Matching in a Single Pass II: Greedy is Optimal},
  author={Assadi, Sepehr and Jiang, Max and Xiang, Mars},
  journal={arXiv preprint arXiv:2607.14656},
  year={2026}
}

@InProceedings{konrad_et_al:LIPIcs.ESA.2026.120,
  author =	{Konrad, Christian and Naidu, Kheeran K. and Walton, Archie and Wang, Eric},
  title =	{On the {Assadi–Liu–Tarjan} Auction Algorithm for Bipartite Matching: Simplification, Alternative Analysis, and Hard Instance},
  booktitle =	{Proceedings of the 34th Annual European Symposium on Algorithms (ESA 2026)},
  pages =	{120:1--120:13},
  year =	{2026},
  volume =	{388},
  doi =		{10.4230/LIPIcs.ESA.2026.120},
}

@article{paz20182,
  title={A {$(2+ \varepsilon)$}-approximation for maximum weight matching in the semi-streaming model},
  author={Paz, Ami and Schwartzman, Gregory},
  journal={ACM Transactions on Algorithms (TALG)},
  volume={15},
  number={2},
  pages={1--15},
  year={2018},
  publisher={ACM New York, NY, USA},
  note={Announced at SODA'17.}
}

@InProceedings{ghaffari_et_al:OASIcs.SOSA.2019.13,
  author =	{Ghaffari, Mohsen and Wajc, David},
  title =	{Simplified and Space-Optimal Semi-Streaming {$(2+ \varepsilon)$}-Approximate Matching},
  booktitle =	{Proceedings of the 2nd Symposium on Simplicity in Algorithms (SOSA 2019)},
  pages =	{13:1--13:8},
  ISBN =	{978-3-95977-099-6},
  ISSN =	{2190-6807},
  year =	{2019},
  volume =	{69},
  doi =		{10.4230/OASIcs.SOSA.2019.13},
}

@article{zelke2012weighted,
  title={Weighted matching in the semi-streaming model},
  author={Zelke, Mariano},
  journal={Algorithmica},
  volume={62},
  number={1},
  pages={1--20},
  year={2012},
  publisher={Springer},
  note={Announed at STACS'08.}
}

@article{epstein2011improved,
  title={Improved approximation guarantees for weighted matching in the semi-streaming model},
  author={Epstein, Leah and Levin, Asaf and Mestre, Juli{\'a}n and Segev, Danny},
  journal={SIAM Journal on Discrete Mathematics},
  volume={25},
  number={3},
  pages={1251--1265},
  year={2011},
  publisher={SIAM},
  note={Announced at STACS'10.}
}

@inproceedings{crouch2014improved,
  title={Improved streaming algorithms for weighted matching, via unweighted matching},
  author={Crouch, Michael and Stubbs, Daniel M},
  booktitle={Proceedings of the Approximation, Randomization, and Combinatorial Optimization. Algorithms and Techniques (APPROX/RANDOM 2014)},
  year={2014},
  organization={Schloss Dagstuhl-Leibniz-Zentrum f{\"u}r Informatik}
}

@inproceedings{assadi2023hidden,
  title={Hidden permutations to the rescue: Multi-pass streaming lower bounds for approximate matchings},
  author={Assadi, Sepehr and Sundaresan, Janani},
  booktitle={Proceedings of the IEEE 64th Annual Symposium on Foundations of Computer Science (FOCS)},
  pages={909--932},
  year={2023},
  organization={IEEE}
}

@article{sundaresan2026new,
  title={New Techniques for Lower Bounds in Graph Streaming and Distributed Computing},
  author={Sundaresan, Janani},
  year={2026},
  publisher={University of Waterloo}
}

@book{schrijver2003combinatorial,
  title = {Combinatorial optimization: polyhedra and efficiency},
  author = {Schrijver, Alexander},
  volume = {24},
  year = {2003},
  publisher = {Springer},
}

\end{document}